\documentclass[11pt]{article}

\usepackage[letterpaper,margin=0.92in]{geometry}
\usepackage[T1]{fontenc}
\usepackage{lmodern}
\usepackage{microtype}
\usepackage{amsmath,amssymb,amsthm,mathtools}
\usepackage{booktabs,array}
\usepackage{float}
\usepackage{tikz}
\usetikzlibrary{arrows.meta,calc,positioning}
\usepackage[colorlinks=true,allcolors=blue]{hyperref}

\newtheorem{theorem}{Theorem}
\newtheorem{proposition}[theorem]{Proposition}
\newtheorem{lemma}[theorem]{Lemma}
\newtheorem{corollary}[theorem]{Corollary}
\newtheorem{question}{Question}

\newtheorem{openquestion}{Open question}

\theoremstyle{definition}

\theoremstyle{remark}
\newtheorem{remark}[theorem]{Remark}

\newcommand{\C}{\mathbb C}
\newcommand{\T}{\{+1,-1\}}
\newcommand{\ket}[1]{\lvert #1\rangle}
\newcommand{\bra}[1]{\langle #1\rvert}
\newcommand{\proj}[1]{\lvert #1\rangle\!\langle #1\rvert}
\newcommand{\ip}[2]{\langle #1\mid #2\rangle}
\newcommand{\Tr}{\operatorname{Tr}}
\newcommand{\supp}{\operatorname{supp}}
\newcommand{\MS}{\mathrm{MS}}
\newcommand{\CalA}{\mathrm{Cal}_{A}}
\newcommand{\CalB}{\mathrm{Cal}_{B}}
\newcommand{\Tw}{\mathrm{CHSH'}}

\newcommand{\cans}[1]{\mathfrak c_{#1}}
\newcommand{\rans}[1]{\mathfrak r_{#1}}

\allowdisplaybreaks

\title{The right entanglement for perfect nonlocality and zero-error communication}
\author{Laura Man\v{c}inska\\[0.3em]\small Department of Mathematical Sciences, University of Copenhagen}
\date{}

\begin{document}
\maketitle

\begin{abstract}    
We identify a structural mechanism by which perfect play in a nonlocal game requires and self-tests a specific non-maximally entangled state. To this end, we construct a Flagged Magic-Square game by extending the standard Magic-Square game with an additional ``flag'' outcome and two additional questions for each player. The game admits a perfect strategy using a state with Schmidt coefficients proportional to $(1,\sqrt{2},\sqrt{2},\sqrt{2},\sqrt{2})$. We give a direct algebraic proof that maximally entangled strategies have winning probability bounded away from one, independently of dimension. Moreover, the game self-tests this state, providing an example of a two-party pseudo-telepathy game that certifies a specific non-maximally entangled state. Our construction illustrates how adding suitable constraints to a rigid measurement geometry can enforce specific Schmidt ratios, suggesting a modular route to further self-tests based on other rigid pseudo-telepathy games.

We establish a general game-to-channel reduction that transfers separations in the game model to one-shot zero-error communication over classical channels. Applied to the Flagged Magic-Square game, it yields a classical channel for which entanglement permits the zero-error transmission of ten messages in a single use, whereas any finite-dimensional maximally entangled resource permits at most nine. Thus, maximally entangled states are not a universally optimal entanglement resource for zero-error communication.

\end{abstract}

\tableofcontents
\newpage
\section{Introduction}\label{sec:Intro}

The phenomenon of ``more nonlocality with less entanglement'' has long illustrated that maximally entangled states need not be optimal for nonlocal tasks~\cite{AcinGillGisin2005,VidickWehner2011}. For a long time, however, it remained unclear whether such a separation could persist in the regime of perfect completeness. In particular, could a two-party pseudo-telepathy game~\cite{BrassardBroadbentTapp2005} be won with certainty using entanglement, but not using a maximally entangled state of any dimension~\cite{Mancinska2014}? Recent work has answered this question and shown that non-maximally entangled states can support pseudo-telepathy~\cite{Renner}, and, moreover, that such states can be necessary for perfect play~\cite{Lalonde}.

The existence of such games raises further questions about what perfect play can enforce. Our starting point is the search for a \emph{natural} construction that isolates the mechanism making non-maximal entanglement necessary for perfect play and provides a template for further applications. We then proceed to ask whether perfect play can certify a particular non-maximal entanglement spectrum, how to construct such games systematically, and whether the resulting resource requirements have consequences for zero-error communication over classical channels. We organize our investigation around the following three questions.

\begin{question}\label{q:self-testing}
Can a two-party pseudo-telepathy game self-test a non-maximally entangled state?
\end{question}
Although every finite-dimensional pure bipartite entangled state can be self-tested~\cite{ColadangeloGohScarani2017}, self-tests based on perfect play in nonlocal games can offer additional advantages. One example is parallel repetition, which provides a natural route to certifying tensor powers while preserving perfect completeness. Indeed, suppose perfect play in a game $G$ self-tests a state $|\psi\rangle$ together with rank-one projective measurements. Winning its $n$-fold parallel repetition with certainty forces each coordinate to attain the optimal single-instance score. The parallelization theorem of \v{S}upi\'c, Cavalcanti, and Bowles~\cite[Theorem~4]{SCB21} then implies that $G^{\otimes n}$ self-tests $|\psi\rangle^{\otimes n}$. Thus, a single such game yields perfectly complete self-tests of arbitrarily many copies of the certified state.

\begin{question}\label{q:construction}
What structural mechanism forces a particular non-maximal entanglement spectrum, and how can it be used to construct self-tests?
\end{question}
There is a basic obstruction to self-testing every pure bipartite entangled state through pseudo-telepathy: finite games specify only countably many perfect-winning conditions, whereas pure bipartite entangled states have uncountably many inequivalent Schmidt spectra. The goal is therefore to develop a general construction principle and understand which entanglement spectra it can certify.

\begin{question}\label{q:zero-error}
Are maximally entangled states sufficient for optimal one-shot zero-error communication over noisy classical channels?
\end{question}
Historically, zero-error communication has long motivated the question of whether maximally entangled states suffice for perfect play. Indeed, for a fixed classical channel and a target number of messages, zero-error encoding and decoding strategies (see Figure~\ref{fig:zero-error-setting}) correspond to perfect strategies for a particular nonlocal game~\cite{CubittLeungMatthewsWinter}. However, the recent separation between maximally and non-maximally entangled resources~\cite{Lalonde} does not, by itself, establish a separation within this particular class of games or answer \autoref{q:zero-error}.

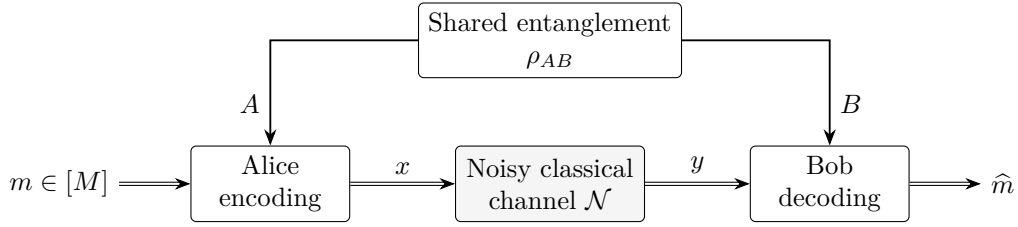
\begin{figure}[htbp]
\centering
\begin{tikzpicture}[
    >=Stealth,
    font=\small,
    operation/.style={draw, rounded corners=2pt, align=center,
        minimum width=2.1cm, minimum height=1cm},
    classical/.style={->, double, double distance=0.8pt, line width=0.4pt},
    quantum/.style={->, line width=0.7pt}
]
\node[operation] (alice) at (-3.7,0) {Alice\\encoding};
\node[operation, minimum width=2.5cm, fill=black!4] (channel) at (0,0)
    {Noisy classical\\channel $\mathcal N$};
\node[operation] (bob) at (3.7,0) {Bob\\decoding};
\node[operation, minimum width=2.8cm, minimum height=0.85cm]
    (resource) at (0,1.9) {Shared entanglement\\$\rho_{AB}$};
\node[anchor=east] at (-5.7,0) {$m\in[M]$};
\node[anchor=west] at (5.7,0) {$\widehat m$};
\draw[classical] (-5.7,0) -- (alice.west);
\draw[classical] (alice.east) -- node[above] {$x$} (channel.west);
\draw[classical] (channel.east) -- node[above] {$y$} (bob.west);
\draw[classical] (bob.east) -- (5.7,0);
\draw[quantum] (resource.west) -| node[pos=0.8,left] {$A$} (alice.north);
\draw[quantum] (resource.east) -| node[pos=0.8,right] {$B$} (bob.north);
\end{tikzpicture}
\caption{One-shot entanglement-assisted zero-error communication~\cite{CubittLeungMatthewsWinter}. Using their shares of $\rho_{AB}$, Alice encodes $m$ as a channel input $x$, and Bob decodes the output $y$. Zero error requires $\widehat m=m$ with certainty for every $m\in[M]$. Double lines carry classical data, and single lines supply the quantum shares.}
\label{fig:zero-error-setting}
\end{figure}

\subsection{Our contributions and ideas}\label{sec:contributions}

We give a transparent and modular construction to explain and certify the necessity of non-maximal entanglement for perfect play. We also establish a general game-to-channel reduction that transfers such entanglement requirements to one-shot zero-error communication over classical channels. Below, we describe the individual contributions, highlighting the key ideas and their relation to the three questions above.

\paragraph{A natural example with an explicit impossibility proof: Flagged Magic-Square game.}
Many of our contributions stem from the introduction of the Flagged Magic-Square game $G_{\mathrm{FMS}}$, which extends the familiar Magic-Square game by adding one additional ``flag'' answer and two additional ``bridge'' questions for each player. Our game admits a perfect strategy with Schmidt coefficients proportional to
\[
    (1,\sqrt{2},\sqrt{2},\sqrt{2},\sqrt{2}),
\]
where the uniform Schmidt spectrum over the last four dimensions corresponds to the two EPR pairs needed for the usual Magic-Square game. Our game provides a natural counterpart to the first reported example, obtained by computer search in~\cite{Lalonde}.

In Theorem~\ref{thm:FMS-separation} we show that
\[
    \omega^*(G_{\mathrm{FMS}})=1>
    \omega_{\mathrm{ME}}(G_{\mathrm{FMS}}),
\]
where $\omega_{\mathrm{ME}}$ is the supremum winning probability over all strategies using a maximally entangled state of arbitrary dimension. We give a direct algebraic proof that no maximally entangled state can enable perfect play. This complements the impossibility proof in~\cite{Lalonde}, which was based on a computer-generated certificate from the tracial NPA hierarchy.

The core idea is as follows. We start with a rigid pseudo-telepathy game, in our case Magic-Square. The calibration constraints (see Section~\ref{sec:game}) link eight selected outcomes of the added bridge questions to four distinguished Magic-Square measurement directions $c_0,c_1,r_0,r_1$ whose cross-overlaps satisfy
\begin{equation}
    \langle c_x \mid r_y \rangle
    =
    \frac{(-1)^{xy}}{\sqrt{2}}.
\label{eq:overview-overlaps}
\end{equation}
More precisely, the game constraints force the eight selected bridge measurement directions to align with the balanced superpositions
\begin{equation*}
    \ket{t_x^a}:=\frac{\ket\bot+a\ket{c_x}}{\sqrt2},
    \qquad
    \ket{u_y^b}:=\frac{\ket\bot+b\ket{r_y}}{\sqrt2},
\end{equation*}
where $x,y\in\{0,1\}$ index the bridge questions and $a,b\in\{-1,+1\}$ label their selected outcomes. In our game, the answer pair $(a,b)\in\{-1,+1\}^2$ is defined to be losing exactly when $ab = - (-1)^{xy}$. Note that
\[\langle t_x^a \mid u_y^b\rangle = \frac{1}{2} (1 + ab \langle c_x \mid r_y\rangle)
= \frac{1}{2} \left(1 + \frac{ab(-1)^{xy}}{\sqrt 2}\right).\]
 Hence, the vectors corresponding to losing sign pairs are not orthogonal in the ordinary inner product, which is precisely what obstructs perfect play with a maximally entangled state. By contrast, for a five-dimensional state with Schmidt coefficients proportional to $(1,\sqrt{2},\sqrt{2},\sqrt{2},\sqrt{2})$, the corresponding \emph{weighted} inner products vanish:
\[
    \left(
        1+\textcolor{red}{\sqrt{2}}\,ab\,\langle c_x \mid r_y\rangle
    \right)
    =
    0
\]
for every losing sign pair. Thus our construction has a simple interpretation: rigidity fixes an overlap of magnitude $1/\sqrt{2}$, while the Schmidt ratio $\sqrt{2}$ is precisely what is needed to cancel it.

\paragraph{A modular construction for certifying entangled states via pseudo-telepathy games.}
In Theorem~\ref{thm:self-test} we show that the Flagged Magic-Square game exactly self-tests both the state with Schmidt coefficients proportional to $(1,\sqrt2,\sqrt2,\sqrt2,\sqrt2)$ and the measurements of the strategy in Section~\ref{sec:honest-strategy}. This answers \autoref{q:self-testing} affirmatively. 

The same mechanism yields a modular construction, addressing \autoref{q:construction}. Magic-Square can be replaced by any \emph{suitable} pseudo-telepathy game $G$ that exactly self-tests a projective strategy using a maximally entangled state of local dimension $d$. The canonical measurements must contain rank-one outcomes along unit vectors $\ket{v_0},\ket{v_1}$ for Alice and the complex conjugates of $\ket{w_0},\ket{w_1}$ for Bob. The two directions on either side need not belong to the same measurement, but their phases must be chosen so that, for some $c>0$,
\begin{equation}
    \ip{v_x}{w_y}=c\,s_{xy},
    \qquad
    s_{xy}=(-1)^{xy},
    \qquad
    x,y\in\{0,1\}.
\label{eq:general-flagging-overlaps}
\end{equation}
For finite-dimensional perfect strategies with a pure shared state, the flagged game then exactly self-tests a state of local dimension $d+1$ with Schmidt coefficients proportional to $(1,c^{-1},\ldots,c^{-1})$, where $c^{-1}$ occurs $d$ times (see Appendix~\ref{app:flagged-base-games}).

For example, we can take the rigid game associated with the Kochen--Specker set CK-31~\cite{CK31}. Its canonical strategy includes the measurement directions
\[
    \ket{v_0}=\frac{(2,-1,0)}{\sqrt5},\qquad
    \ket{v_1}=\frac{(0,2,1)}{\sqrt5},\qquad
    \ket{w_0}=\frac{(1,0,2)}{\sqrt5},\qquad
    \ket{w_1}=\frac{(1,0,-2)}{\sqrt5},
\]
which satisfy~\eqref{eq:general-flagging-overlaps} with $c=2/5$\footnote{Other choices of vectors from CK-31 give additional choices for $c$.}. The flagged game therefore self-tests a non-maximally entangled state of local dimension four with Schmidt coefficients proportional to $(2,5,5,5)$. Our construction cannot yield a local-dimension-three example: it adds one flag dimension, and no two-player pseudo-telepathy game exists in local dimension two~\cite{BrassardMethotTapp}.

\textit{Concurrent related work.} While preparing this manuscript, we became aware of concurrent work by Ranyiliu Chen~\cite{Chen2026}. Starting from Lalonde's six-dimensional inner product game~\cite{Lalonde}, Chen constructs a reduced game which \emph{robustly} self-tests the state with Schmidt coefficients proportional to $(1,2,2,2,2)$. Thus his work also answers \autoref{q:self-testing} affirmatively.

\paragraph{From perfect nonlocality to optimal zero-error communication.}
Our game-to-channel reduction transfers perfect-play resource separations to zero-error communication. Given a two-party game $G$, we construct a channel whose confusability graph has local question--answer pairs as vertices. Answers to the same question form a clique, while cross-player edges join answers that lose on the corresponding question pair. Theorem~\ref{thm:game-channel-transfer} establishes that $G$ has a perfect strategy if and only if the channel admits an entanglement-assisted zero-error code transmitting $q=|X|+|Y|$ messages in one use. Here $X$ and $Y$ are the question sets for Alice and Bob, respectively.

More precisely, the correspondence relates perfect game strategies using the shared state $\ket{\psi}$ to $q$-message zero-error codes using $\ket{\chi}$, where
\[
    \ket{\psi}=\sum_r\lambda_r\ket{r,r},
    \qquad
    \ket{\chi}=
    \frac{\sum_r\sqrt{\lambda_r}\ket{r,r}}
         {\sqrt{\sum_r\lambda_r}}.
\]
Thus the code's Schmidt coefficients are the normalized square roots of the game's, preserving Schmidt rank and maximal entanglement. 

For Flagged Magic-Square, this gives a $50$-input classical channel that transmits $10$ messages perfectly with a suitable non-maximally entangled state, whereas any maximally entangled resource permits at most $9$ (Corollary~\ref{cor:FMS-zero-error}). Thus non-maximal entanglement can be necessary for optimal one-shot zero-error communication, answering \autoref{q:zero-error} negatively.

\subsection{Open questions}
\label{sec:discussion}
Our results open several directions for further investigation.

\begin{samepage}
\begin{openquestion}\label{oq:separation-amplification}
Does there exist a family of games $G_n$, each admitting a perfect finite-dimensional strategy, such that
\[
    \omega^*(G_n)=1
    \qquad\text{and}\qquad
    \omega_{\mathrm{ME}}(G_n)\longrightarrow 0?
\]
\end{openquestion}
\end{samepage}
A natural candidate is the parallel repetition $G_{\mathrm{FMS}}^{\otimes n}$, which preserves perfect completeness. However, it is not known whether parallel repetition drives the maximally entangled value of $G_{\mathrm{FMS}}$ to zero. Existing parallel-repetition results assume $\omega^*(G)<1$~\cite{YuenParallel}, and therefore do not settle this question.

Interestingly, for any fixed game $G$ admitting a perfect finite-dimensional strategy, $\omega_{\mathrm{ME}}(G^{\otimes n})$ cannot decay exponentially in $n$. Indeed, let $\ket{\psi}$ be the shared state of a perfect strategy for $G$, with Schmidt rank $d$. Grouping the Schmidt strings of $\ket{\psi}^{\otimes n}$ by type gives
\[
    \ket{\psi}^{\otimes n}
      =\sum_t\sqrt{p_t}\,\ket{\Phi_t},
\]
with at most $\binom{n+d-1}{d-1}$ terms, each maximally entangled on its support. The game operator $0\leq W\leq I$ of the product measurements has $\ket{\psi}^{\otimes n}$ as a $(+1)$-eigenvector. This implies
\[
    \bra{\Phi_t}W\ket{\Phi_t}
      \geq
      \bigl|\langle\Phi_t\mid\psi^{\otimes n}\rangle\bigr|^2
      =p_t.
\]
Choosing the largest-weight term therefore yields $\omega_{\mathrm{ME}}(G^{\otimes n})\geq\binom{n+d-1}{d-1}^{-1}$, and any decay to zero must respect this polynomial lower bound.

\begin{openquestion}\label{oq:schmidt-spectra}
Which Schmidt spectra can be self-tested by pseudo-telepathy games?
\end{openquestion}
Our previous counting argument limits the Schmidt spectra and ratios of self-tested states to a countable set. It would be interesting to identify other nontrivial structural restrictions on the attainable spectra and systematic ways to realize them. Our construction translates overlaps in rigid measurement geometries into certified Schmidt ratios, making the range of attainable overlaps a natural starting point.

\begin{openquestion}\label{oq:minimal-schmidt-rank}
What is the smallest Schmidt rank of a pure state that perfectly wins a two-party game $G$ with $\omega_{\mathrm{ME}}(G)<1$?
\end{openquestion}
The flagged construction based on the CK-31 Kochen--Specker game shows that Schmidt rank four is possible. So the question is whether such a separation can already occur with a Schmidt-rank-three state.

\section{Notation}\label{sec:notation}
\subsection{Entangled, maximally entangled, and tracial values}

For a nonlocal game $G$, let $\omega^*(G)$ be the supremum of the winning probability over finite-dimensional tensor-product strategies. The maximally entangled value $\omega_{\mathrm{ME}}(G)$ restricts this supremum to strategies using a shared pure state whose Schmidt coefficients are all equal and nonzero.

The unrestricted tracial value $\omega_{\mathrm{tr}}(G)$ is a commuting-operator analogue of the maximally entangled value. It is the supremum of the winning probability over correlations of the form
\[
p(a,b\mid x,y)=\tau(E_a^xF_b^y),
\]
where $\tau$ is a tracial state on a unital $C^*$-algebra and $\{E_a^x\},\{F_b^y\}$ are POVMs. No finite-dimensionality assumption is imposed on the algebra or on the GNS representation. Maximally entangled strategies of local dimension $d$ have this form with $\tau=\operatorname{Tr}/d$ after transposing Bob's effects in the Schmidt basis. For general tracial states, left and right multiplication in the GNS representation realize these correlations using commuting local measurements~\cite{TodorovTurowska}. Consequently,
\[
\omega_{\mathrm{ME}}(G)\le \omega_{\mathrm{tr}}(G).
\]
The tracial game value is attained by compactness of the tracial correlation class (see the arguments in~\cite{DykemaPaulsen} for the synchronous setting and~\cite{TodorovTurowska} for general POVMs).

\subsection{Magic-Square notation}\label{sec:MS-notation}

The Magic-Square game has several variants. We fix the one used here and introduce notation that will be useful for the flagged construction.

We use the following Mermin--Peres square on two qubits~\cite{Mermin,Peres}:
\begin{equation}
\begin{array}{c|ccc}
 & C_1&C_2&C_3\\ \hline
R_1& Z\otimes I&I\otimes Z&Z\otimes Z\\
R_2& I\otimes X&X\otimes I&X\otimes X\\
R_3& Z\otimes X&X\otimes Z&Y\otimes Y.
\end{array}
\label{eq:magic-square}
\end{equation}
The observables in each row or column commute.  Every row and column product is $+I$, except for the third column, whose product is $-I$.  Write
\begin{equation}
\kappa_1=\kappa_2=+1,
\qquad
\kappa_3=-1
\label{eq:column-parities}
\end{equation}
for the column parities.

In the Magic-Square game, Alice receives a row question $R_i$ and Bob receives a column question $C_j$. Their valid answers are sign triples $s,t\in\T^3$ satisfying $s_1s_2s_3=+1$ and $t_1t_2t_3=\kappa_j$, respectively. They win exactly when $s_j=t_i$.

We distinguish two answer strings for $R_1$ and two separately labeled answer strings for $C_1$:
\begin{equation}
\begin{aligned}
\rans{0}&=(+1,+1,+1),
&\rans{1}&=(+1,-1,-1),\\
\cans{0}&=(+1,+1,+1),
&\cans{1}&=(+1,-1,-1).
\end{aligned}
\label{eq:distinguished-answers}
\end{equation}
The Fraktur symbols emphasize that these are answer strings.  In the strategy defined by \eqref{eq:magic-square}, they correspond to the vectors
\begin{equation}
\ket{r_0}=\ket{00},
\quad
\ket{r_1}=\ket{01},
\qquad
\ket{c_0}=\ket{0+},
\quad
\ket{c_1}=\ket{0-}.
\label{eq:distinguished-vectors}
\end{equation}
Their overlap matrix is
\begin{equation}
\sqrt2\,\ip{c_x}{r_y}=(-1)^{xy},
\qquad x,y\in\{0,1\}.
\label{eq:distinguished-overlaps}
\end{equation}

\subsection{One-shot zero-error communication}
\label{sec:zero-error-preliminaries}

A finite classical channel $\mathcal N$ is specified by conditional probabilities $\mathcal N(o\mid v)$ of an output $o$ given an input $v$. Its \emph{confusability graph} $H$ has the channel inputs as vertices. Distinct inputs $v,w$ are adjacent exactly when $\mathcal N(o\mid v)\mathcal N(o\mid w)>0$ for some output $o$. For any graph $H$, write $v\simeq w$ if $v=w$ or $vw\in E(H)$.

In a one-shot entanglement-assisted code, Alice and Bob share a finite-dimensional bipartite state. To send a message $i\in\{1,\ldots,M\}$, Alice performs an encoding POVM whose outcome $v$ is used as the channel input. Bob chooses a decoding measurement according to the channel output and must recover $i$ with certainty for every message. The only communication is this single channel use.

By assigning a purification to one player if necessary, we may assume the shared state is pure. Throughout the zero-error communication setting, we take the local Hilbert spaces to be its Schmidt supports. Thus all Schmidt coefficients are positive and Bob's reduced state $\rho$ is positive definite.

We use the operator characterization of such codes from \cite{CubittLeungMatthewsWinter}. An $M$-message zero-error code exists if and only if there are a positive definite density operator $\rho$ and positive semidefinite operators $\{\beta_v^i\}$ on a common finite-dimensional Hilbert space satisfying
\begin{equation}
    \sum_{v\in V(H)}\beta_v^i=\rho
    \quad (1\leq i\leq M),
    \qquad
    \operatorname{Tr}(\beta_v^i\beta_w^j)=0
    \quad (i\neq j,\ v\simeq w).
    \label{eq:zero-error-code}
\end{equation}
Here $\beta_v^i$ is Bob's unnormalized conditional state when Alice encodes message $i$ using input $v$, and $\rho$ is his reduced state before encoding. The first condition expresses that this reduced state is independent of the message. The second requires orthogonal supports for distinct messages whenever their inputs can produce the same output. The largest feasible $M$ is the \emph{entanglement-assisted independence number} $\alpha^*(H)$.

This characterization shows that the maximum number of messages depends only on $H$, both for unrestricted and for maximally entangled assistance~\cite{CubittLeungMatthewsWinter}. Conversely, every finite graph is the confusability graph of a classical channel. For example, let each nonisolated input vertex output a uniformly random incident edge, and assign each isolated vertex its own deterministic output symbol.

For completeness, we describe how the operators $\{\beta_v^i\}$ give rise to an entanglement-assisted zero-error code. Write $\rho=\sum_{r=1}^d p_r|r\rangle\langle r|$ with $p_r>0$. A suitable shared state and encoding POVMs are
\[
    |\chi_\rho\rangle
      =\sum_{r=1}^d\sqrt{p_r}\,|r,r\rangle,
    \qquad
    M_v^i
      =\bigl(\rho^{-1/2}\beta_v^i\rho^{-1/2}\bigr)^T.
\]
Here $T$ denotes transpose in the chosen eigenbasis. The conditional states are $\rho^{1/2}(M_v^i)^T\rho^{1/2}=\beta_v^i$. For each output $o$, Bob's unnormalized conditional states $\tau_i^o=\sum_v\mathcal N(o\mid v)\beta_v^i$ have mutually orthogonal supports by~\eqref{eq:zero-error-code}, so he can distinguish the messages perfectly.

We write $\alpha^*_{\mathrm{ME}}(H)$ for the same maximum when the shared state is required to be maximally entangled, $|\Phi_d\rangle=d^{-1/2}\sum_{r=1}^d|r,r\rangle$, for some finite $d$. Equivalently, impose $\rho=I_d/d$ in \eqref{eq:zero-error-code}. In this case $\beta_v^i=(M_v^i)^T/d$, so the encoding POVMs satisfy
\begin{equation}
    \operatorname{Tr}(M_v^iM_w^j)=0
    \qquad (i\neq j,\ v\simeq w).
    \label{eq:me-zero-error-code}
\end{equation}
The encoding POVMs are arbitrary and need not be projective.

\section{Flagged Magic-Square game}\label{sec:FMS-game}

\subsection{Questions, answers, and verification}\label{sec:game}

\paragraph{Question sets.}
The verifier samples uniformly from $X\times Y$, where
\[
X=\{R_1,R_2,R_3,T_0,T_1\},
\qquad
Y=\{C_1,C_2,C_3,U_0,U_1\}.
\]
The questions $R_i$ and $C_j$ are the \emph{Magic-Square questions}. The two additional questions for each player, $T_0,T_1$ for Alice and $U_0,U_1$ for Bob, are called \emph{bridge questions}. Uniformity is inessential for perfect play and merely fixes a concrete game distribution.

\paragraph{Answer sets.}
Each question has five answers. A row or column question has its four valid Magic-Square answers, as defined in Section~\ref{sec:MS-notation}, together with $\bot$. For $x,y\in\{0,1\}$, the bridge question $T_x$ has the two \emph{sign answers} $\pm1$ and the three valid Magic-Square answers to $C_1$ other than $\cans{x}$. Similarly, $U_y$ uses the two sign answers and the three valid Magic-Square answers to $R_1$ other than $\rans{y}$. The Magic-Square answers retain their $C_1$ or $R_1$ labels.

\paragraph{Verification predicate.}
The verification rules are summarized by
\begin{equation}
\renewcommand{\arraystretch}{1.35}
\begin{array}{c|ccccc}
 & C_1&C_2&C_3&U_0&U_1\\ \hline
R_1&\MS&\MS&\MS&\CalB(0)&\CalB(1)\\
R_2&\MS&\MS&\MS&*&*\\
R_3&\MS&\MS&\MS&*&*\\
T_0&\CalA(0)&*&*&\Tw(0,0)&\Tw(0,1)\\
T_1&\CalA(1)&*&*&\Tw(1,0)&\Tw(1,1).
\end{array}
\label{eq:question-table}
\end{equation}
We first explain the \emph{idea} behind each of the constraints, followed by the formal definitions. In the Magic-Square part, the players synchronize their flag answers and give consistent Magic-Square answers otherwise. This equips each player with a fixed flag subspace with rigid Magic-Square measurements on its complement. The \emph{calibration constraints} confine each bridge measurement's sign outcomes to the flag and one selected Magic-Square outcome. The CHSH-like \emph{frustrated sign constraints} then force the bridge measurements to mix these contributions so that losing amplitudes cancel. This cancellation requires unequal Schmidt coefficients, ruling out perfect play with a maximally entangled state.

We now state the rules formally.

\paragraph{Magic-Square rule $\MS$.}
For $(R_i,C_j)$:
\begin{itemize}
\item $(\bot,\bot)$ wins.
\item exactly one flag answer $\bot$ loses.
\item two nonflag answers $s,t$ win exactly when $s_j=t_i$.
\end{itemize}

\paragraph{Alice calibration constraint $\CalA(x)$.}
For $(T_x,C_1)$, Alice's answer $a$ and Bob's answer $b$ win if and only if one of the following conditions holds:
\begin{enumerate}
\item $a\in\T$ and $b\in\{\bot,\cans{x}\}$.
\item $a$ is a Magic-Square answer and $a=b$.
\end{enumerate}
Every other answer pair loses.

\paragraph{Bob calibration constraint $\CalB(y)$.}
For $(R_1,U_y)$, Alice's answer $a$ and Bob's answer $b$ win if and only if one of the following conditions holds:
\begin{enumerate}
\item $b\in\T$ and $a\in\{\bot,\rans{y}\}$.
\item $b$ is a Magic-Square answer and $b=a$.
\end{enumerate}
Every other answer pair loses.

\paragraph{Frustrated sign rule $\Tw(x,y)$.}
For $(T_x,U_y)$:
\begin{itemize}
\item if both answers are signs $a,b\in\T$, they win exactly when
\begin{equation}
 ab=(-1)^{xy}.
 \label{eq:twisted-condition}
\end{equation}
\item if at least one answer is a Magic-Square answer, they win automatically.
\end{itemize}

\paragraph{Free rule $*$.}
Every answer pair wins.  Equivalently, these question pairs may be removed from the support of the game distribution.

\subsection{A perfect nonmaximally entangled strategy}\label{sec:honest-strategy}

Let the local space be
\[
\mathcal H=\C\ket\bot\oplus(\C^2\otimes\C^2),
\qquad
D_0=1\oplus\sqrt2 I_4.
\]
The shared state is
\begin{equation}
\ket\psi
 =\frac{\operatorname{vec}(D_0)}{3}
 =\frac13\ket{\bot,\bot}
 +\frac{\sqrt2}{3}\sum_{z\in\{0,1\}^2}\ket{z,z}.
\label{eq:honest-state}
\end{equation}

For $R_i$, Alice measures in the orthonormal basis consisting of $\ket\bot$ and a joint eigenbasis of row $i$ in \eqref{eq:magic-square}, with the natural answer labels.  For $C_j$, Bob uses the complex conjugates of the corresponding column eigenvectors, together with $\ket\bot$.  On the four-dimensional summand this is the usual perfect Magic-Square strategy, while the one-dimensional summand gives a common flag outcome.

For $T_x$, Alice uses the two vectors
\begin{equation}
\ket{t_x^a}
 :=\frac{\ket\bot+a\ket{c_x}}{\sqrt2},
\qquad a\in\T,
\label{eq:T-basis}
\end{equation}
together with the other three vectors of the $C_1$ eigenbasis.  The vectors in \eqref{eq:T-basis} receive sign labels, and the remaining vectors retain their $C_1$ answer labels.  For $U_y$, Bob uses the complex conjugates of
\begin{equation}
\ket{u_y^b}
 :=\frac{\ket\bot+b\ket{r_y}}{\sqrt2},
\qquad b\in\T,
\label{eq:U-basis}
\end{equation}
together with the other three vectors of the $R_1$ eigenbasis.

\begin{proposition}\label{prop:perfect}
The state and measurements above win $G_{\mathrm{FMS}}$ with probability one.
\end{proposition}

\begin{proof}
For unit vectors $v,w\in\mathcal H$, the joint amplitude for Alice's outcome associated with $v$ and Bob's outcome associated with $\overline w$ is proportional to
\[
\bra vD_0\ket w.
\]
On the Magic-Square block, mixed flag/nonflag amplitudes vanish because the two summands of $\mathcal H$ are orthogonal.  If two nonflag outcomes disagree at the tested cell, their eigenvectors have opposite eigenvalues for the same Hermitian observable and are therefore orthogonal.  Hence every losing $\MS$ pair has zero probability.

For $(T_x,C_1)$, a sign vector $t_x^a$ has support only on $\ket\bot$ and $\ket{c_x}$.  Every remaining $T_x$ basis vector is a $C_1$ eigenvector orthogonal to $\ket{c_x}$ and has nonzero amplitude only with the matching $C_1$ outcome.  This is precisely the support pattern prescribed by $\CalA(x)$.  The verification of $\CalB(y)$ is identical.

For sign outcomes in a frustrated cell,
\begin{align}
\bra{t_x^a}D_0\ket{u_y^b}
 &=\frac12\left(1+\sqrt2\,ab\,\ip{c_x}{r_y}\right)\notag\\
 &=\frac12\left(1+ab(-1)^{xy}\right),
\label{eq:twisted-amplitude}
\end{align}
by \eqref{eq:distinguished-overlaps}.  This vanishes exactly when $ab=-(-1)^{xy}$, which is the losing condition in \eqref{eq:twisted-condition}.
\end{proof}

\subsection{Impossibility of perfect play with maximally entangled states}\label{sec:no-ME}

We use the following exact rigidity statement.  Section~\ref{app:tracial-MS-rigidity} gives a self-contained proof. Robust variants appear in, for example,~\cite{ColadangeloStark,CoudronNatarajan}.

\begin{proposition}[Magic-Square rigidity]\label{prop:tracial-MS-rigidity}
Let $\{A_s^i\}_s$ and $\{B_t^j\}_t$, $i,j=1,2,3$, be POVMs on a Hilbert space $\mathcal K_0$, with outcomes indexed by the valid Magic-Square answers to $R_i$ and $C_j$, respectively (Section~\ref{sec:MS-notation}). Suppose
\begin{equation}
A_s^iB_t^j=0
\qquad\text{whenever }s_j\ne t_i.
\label{eq:tracial-MS-zero}
\end{equation}
Then all six POVMs are projective, and there exist a Hilbert space $\mathcal K$ and a unitary identification
\[
\mathcal K_0\cong\C^2\otimes\C^2\otimes\mathcal K
\]
under which the six measurements are the joint-eigenbasis measurements of \eqref{eq:magic-square}, tensored with $I_{\mathcal K}$.
\end{proposition}

\begin{theorem}[Flagged Magic-Square separation]\label{thm:FMS-separation}
The Flagged Magic-Square game satisfies
\[
\omega^*(G_{\mathrm{FMS}})=1>\omega_{\mathrm{tr}}(G_{\mathrm{FMS}}).
\]
\end{theorem}

\begin{proof}
Proposition~\ref{prop:perfect} gives $\omega^*(G_{\mathrm{FMS}})=1$.  Since the tracial value is attained, it suffices to rule out a perfect tracial strategy. For a contradiction, assume that $G_{\mathrm{FMS}}$ has a perfect tracial strategy.  Thus there are POVMs $\{A_a^q\}_{a,q}$ and $\{B_b^r\}_{b,r}$ in a unital $C^*$-algebra $\mathcal A$ and a faithful tracial state $\tau$ such that
\[
p(a,b\mid q,r)=\tau(A_a^qB_b^r).
\]
Represent $\mathcal A$ faithfully on a Hilbert space $\mathcal H$. The faithfulness of the state implies that every losing answer pair yields a zero product, $A_a^qB_b^r=0$, between the corresponding effects.

\paragraph{Step 1: isolate the two sectors.}
Let $G_i$ be Alice's flag effect for $R_i$, and let $H_j$ be Bob's flag effect for $C_j$.  Mixed flag/nonflag pairs lose, so
\[
G_i(I-H_j)=0,
\qquad
(I-G_i)H_j=0.
\]
It follows that
\[
G_i=G_iH_j=H_j
\]
for all $i,j$.  Denote their common value by $P^\perp$, and set
\[
P:=I-P^\perp.
\]
Thus $P$ is the projection onto the Magic-Square sector and $P^\perp$ is the flag projection.

The four nonflag effects for each row or column are supported on $P$ and sum to $P$.  Within the corner algebra $P\mathcal A P$, they form a perfect ordinary Magic-Square strategy.  Proposition~\ref{prop:tracial-MS-rigidity} gives a multiplicity space $\mathcal K$ and an identification
\begin{equation}
P\mathcal H\cong\C^2\otimes\C^2\otimes\mathcal K
\label{eq:MS-decomp}
\end{equation}
under which the distinguished outcome projections are
\begin{equation}
\Pi_{\cans{x}}=\proj{c_x}\otimes I_{\mathcal K},
\qquad
\Pi_{\rans{y}}=\proj{r_y}\otimes I_{\mathcal K}.
\label{eq:distinguished-projections}
\end{equation}
Consequently,
\begin{equation}
\Pi_{\cans{x}}\Pi_{\rans{y}}
 =\ip{c_x}{r_y}\,\ket{c_x}\!\bra{r_y}\otimes I_{\mathcal K},
\qquad
\bigl\|\Pi_{\cans{x}}\Pi_{\rans{y}}\bigr\|\le\frac1{\sqrt2}.
\label{eq:projection-overlap}
\end{equation}
If $P\ne0$, the inequality in \eqref{eq:projection-overlap} is an equality, and the product is nonzero.  If $P=0$, we interpret $\mathcal K=\{0\}$ and all these projections as zero.

\paragraph{Step 2: calibration constraints determine the bridge supports.}
Let $E_x^+,E_x^-$ be Alice's sign effects for $T_x$, and let $F_y^+,F_y^-$ be Bob's sign effects for $U_y$.  The calibration constraints and Lemma~\ref{lem:component-matching} imply
\begin{equation}
E_x^++E_x^-=P^\perp+\Pi_{\cans{x}}=:S_x^A,
\qquad
F_y^++F_y^-=P^\perp+\Pi_{\rans{y}}=:S_y^B.
\label{eq:bridge-supports}
\end{equation}
Define
\begin{equation}
A_x:=E_x^+-E_x^-,
\qquad
B_y:=F_y^+-F_y^-,
\label{eq:sign-decomp}
\end{equation}
which are self-adjoint contractions supported on $S_x^A$ and $S_y^B$, respectively.

\paragraph{Step 3: the angled supports eliminate the flag sector.}
Fix $x,y\in\{0,1\}$.  The two losing sign pairs satisfy $b=-(-1)^{xy}a$. Expanding $A_xB_y$ and $A_xS_y^B$ using \eqref{eq:sign-decomp} and discarding the two products corresponding to the losing sign pairs yields
\begin{align}
A_xB_y&=(-1)^{xy}S_x^AS_y^B,
\label{eq:core1}\\
A_xS_y^B&=(-1)^{xy}S_x^AB_y.
\label{eq:core2}
\end{align}

Compressing \eqref{eq:core2} by $P^\perp$ gives
\[
P^\perp A_xP^\perp=(-1)^{xy}P^\perp B_yP^\perp.
\]
The four sign cells in the frustrated square give
\[
P^\perp A_0P^\perp
 =P^\perp B_0P^\perp
 =P^\perp B_1P^\perp,
\qquad
P^\perp A_1P^\perp
 =P^\perp B_0P^\perp
 =-P^\perp B_1P^\perp.
\]
Hence
\begin{equation}
P^\perp A_xP^\perp=P^\perp B_yP^\perp=0
\qquad(x,y\in\{0,1\}).
\label{eq:flag-diagonal-zero}
\end{equation}

Compressing \eqref{eq:core1} by $P^\perp$ gives
\begin{equation}
P^\perp A_xB_yP^\perp=(-1)^{xy}P^\perp.
\label{eq:flag-identity}
\end{equation}
On the other hand, the support relations and \eqref{eq:flag-diagonal-zero} give
\[
P^\perp A_x=P^\perp A_x\Pi_{\cans{x}},
\qquad
B_yP^\perp=\Pi_{\rans{y}}B_yP^\perp.
\]
Therefore
\[
P^\perp A_xB_yP^\perp
 =P^\perp A_x\Pi_{\cans{x}}\Pi_{\rans{y}}B_yP^\perp,
\]
and \eqref{eq:projection-overlap} implies
\[
\bigl\|P^\perp A_xB_yP^\perp\bigr\|
 \le \|A_x\|\,
      \bigl\|\Pi_{\cans{x}}\Pi_{\rans{y}}\bigr\|\,
      \|B_y\|
 \le\frac1{\sqrt2}.
\]
If $P^\perp\ne0$, the right-hand side of \eqref{eq:flag-identity} has norm one, a contradiction.  Thus
\begin{equation}
P^\perp=0.
\label{eq:flag-zero}
\end{equation}

\paragraph{Step 4: the frustrated parity eliminates the Magic-Square sector.}
With $P^\perp=0$, equation \eqref{eq:bridge-supports} becomes
\[
E_x^++E_x^-=\Pi_{\cans{x}},
\qquad
F_y^++F_y^-=\Pi_{\rans{y}}.
\]
Every positive operator supported on one of the rank-one tensor projections in \eqref{eq:distinguished-projections} factors through the multiplicity space.  Thus there are POVM effects $\widehat E_x^\pm,\widehat F_y^\pm$ on $\mathcal K$ such that
\begin{align*}
E_x^a&=\proj{c_x}\otimes\widehat E_x^a,
&\widehat E_x^++\widehat E_x^-&=I_{\mathcal K},\\
F_y^b&=\proj{r_y}\otimes\widehat F_y^b,
&\widehat F_y^++\widehat F_y^-&=I_{\mathcal K}.
\end{align*}
For a losing sign pair,
\[
0=E_x^aF_y^b
 =\ip{c_x}{r_y}\,\ket{c_x}\!\bra{r_y}
   \otimes\widehat E_x^a\widehat F_y^b.
\]
Every overlap in \eqref{eq:distinguished-overlaps} is nonzero, so $\widehat E_x^a\widehat F_y^b=0$.  Lemma~\ref{lem:binary-matching} therefore gives
\begin{equation}
\widehat E_x^a=\widehat F_y^{(-1)^{xy}a}.
\label{eq:multiplicity-matching}
\end{equation}
Applying \eqref{eq:multiplicity-matching} to the four cells of the sign square gives
\[
\widehat E_0^+=\widehat F_0^+=\widehat F_1^+,
\qquad
\widehat E_1^+=\widehat F_0^+=\widehat F_1^-.
\]
Hence $\widehat F_1^+=\widehat F_1^-$.  Lemma~\ref{lem:binary-matching} also shows that these are complementary projections, so their equality forces $I_{\mathcal K}=0$.  It follows from \eqref{eq:MS-decomp} that $P=0$.  Together with \eqref{eq:flag-zero}, this gives $I=P+P^\perp=0$, a contradiction.
\end{proof}

\begin{corollary}\label{cor:FMS-ME-gap}
No finite-dimensional maximally entangled strategy wins $G_{\mathrm{FMS}}$ perfectly. Moreover,
\[
\omega_{\mathrm{ME}}(G_{\mathrm{FMS}})\le \omega_{\mathrm{tr}}(G_{\mathrm{FMS}})<1.
\]
\end{corollary}

\section{One-shot zero-error communication}\label{sec:zero-error}

We establish a reduction from finite bipartite nonlocal games to one-shot zero-error communication. Given a game $G$, we construct a classical channel, specified through its confusability graph, that admits a $q$-message entanglement-assisted zero-error code if and only if $G$ has a perfect finite-dimensional strategy, where $q$ is the total number of local questions. The correspondence between strategies and codes preserves Schmidt rank and maximal entanglement, allowing us to transfer resource separations from games to communication.

\subsection{The game-to-channel construction}
\label{sec:game-channel-construction}

Let $G$ be a bipartite game with nonempty finite question sets $X,Y$, fixed nonempty finite answer sets $\mathcal A,\mathcal B$, a question distribution $\pi$ with full support, and a verification predicate $V(a,b\mid x,y)$. Regard $X$ and $Y$ as disjoint sets, and let
\[
    q:=|X|+|Y|.
\]

\paragraph{The channel construction.}
We define the classical noisy channel $\mathcal{N}_G$ by specifying its confusability graph.
Our graph is closely related to the game graphs of Man\v{c}inska, Roberson, and Varvitsiotis~\cite{MancinskaRobersonVarvitsiotis2016}. Define $\Gamma_G$ on the local question--answer pairs
\[
    V(\Gamma_G)
      =(X\times\mathcal A)\sqcup(Y\times\mathcal B).
\]
For $x\in X$ and $y\in Y$, make each answer block $C_x:=\{x\}\times\mathcal A$ and $C_y:=\{y\}\times\mathcal B$ a clique. Join $(x,a)$ to $(y,b)$ exactly when $V(a,b\mid x,y)=0$. There are no other edges and the answer blocks partition $V(\Gamma_G)$ into $q$ cliques.
This $q$-clique partition implies that $\alpha^*(\Gamma_G)\leq q$~\cite{Beigi2010}. We will need the following description of codes attaining this bound. We also include a direct proof of the upper bound itself for completeness.

\begin{lemma}[Saturation of the clique-partition bound]
\label{lem:clique-partition-saturation}
Suppose $V(H)$ is partitioned into $q$ cliques $C_z$. Every code satisfying~\eqref{eq:zero-error-code} has $M\leq q$. If $M=q$, let
\[
    \sigma_{i,z}:=\sum_{v\in C_z}\beta_v^i,
    \qquad
    P_{i,z}:=\supp(\sigma_{i,z}),
\]
where $\supp$ denotes the support projection. Then
\begin{equation}
    \sigma_{i,z}=\rho P_{i,z}=P_{i,z}\rho,
    \qquad
    \sum_iP_{i,z}=I,
    \qquad
    \sum_zP_{i,z}=I.
\label{eq:clique-partition-saturation}
\end{equation}
In particular, the projections in each row and each column of $(P_{i,z})$ are pairwise orthogonal.
\end{lemma}

\begin{proof}
For fixed $z$, the operators $\sigma_{i,z}$ have pairwise orthogonal supports, by~\eqref{eq:zero-error-code} and the clique condition. Moreover, $0\leq\sigma_{i,z}\leq\rho$. Hence
\[
    \sum_i\Tr(\sigma_{i,z})
    \leq\sum_i\Tr(P_{i,z}\rho)
    \leq\Tr\rho=1.
\]
Since $M=\sum_{i,z}\Tr(\sigma_{i,z})$, summing the above over $z$ gives $M\leq q$.

If $M=q$, every inequality above is an equality for every $z$. Since $\rho>0$ and $I-\sum_iP_{i,z}\geq0$, we obtain $\sum_iP_{i,z}=I$. Equality also holds separately for every $i,z$ in $\Tr(\sigma_{i,z})\leq\Tr(P_{i,z}\rho)$, so
\[
    \Tr\bigl(P_{i,z}(\rho-\sigma_{i,z})\bigr)=0
\]
Positivity of $\rho-\sigma_{i,z}$ implies $(\rho-\sigma_{i,z})P_{i,z}=0$. 
Since $P_{i,z}$ is the support projection of $
\sigma_{i,z}$, we get
\[
    \rho P_{i,z}=\sigma_{i,z}=P_{i,z}\rho.
\]
Finally, $\sum_z\sigma_{i,z}=\rho$ together with $\rho$ having full rank implies $\sum_zP_{i,z}=I$. 
\end{proof}

\subsection{Perfect strategies and optimal zero-error codes}
\label{sec:game-channel-correspondence}

The following theorem is the key to transferring entanglement requirements from games to communication. It gives an explicit correspondence between shared states permitting perfect play in $G$ and those permitting $q$-message zero-error transmission through $\mathcal N_G$.

\begin{theorem}[Perfect-strategy correspondence]
\label{thm:game-channel-transfer}
Let $D=\operatorname{diag}(\lambda_1,\ldots,\lambda_d)>0$ satisfy $\Tr(D^2)=1$, and define
\begin{equation}
    \ket{\psi_D}=\sum_{r=1}^d\lambda_r\ket{r,r},
    \qquad
    \ket{\chi_D}=
    \frac{1}{\sqrt{\Tr D}}
    \sum_{r=1}^d\sqrt{\lambda_r}\ket{r,r}.
\label{eq:square-root-code-state}
\end{equation}
Then $G$ has a perfect strategy using $\ket{\psi_D}$ if and only if $\mathcal N_G$ admits a $q$-message one-shot zero-error code using $\ket{\chi_D}$. Both directions preserve Schmidt rank and maximal entanglement.
\end{theorem}

\begin{proof}[Proof of Theorem~\ref{thm:game-channel-transfer}]
\emph{From a perfect strategy to a zero-error code.} Let $\{E_a^x\}$ and $\{F_b^y\}$ be the POVMs of a perfect strategy using $\ket{\psi_D}$. In the Schmidt basis, every losing tuple satisfies
\begin{equation}
    0=\bra{\psi_D}E_a^x\otimes F_b^y\ket{\psi_D}
      =\Tr\bigl(E_a^xD(F_b^y)^TD\bigr).
\label{eq:losing-Schmidt-zero}
\end{equation}
Put $\rho=D/\Tr D$ and index the $q$ messages by $X\sqcup Y$. Define
\[
    \beta_{(x,a)}^x=\rho^{1/2}E_a^x\rho^{1/2},
    \qquad
    \beta_{(y,b)}^y=\rho^{1/2}(F_b^y)^T\rho^{1/2},
\]
with each message family zero outside its own answer block. Each family sums to $\rho$, and every losing tuple satisfies
\[
    \Tr\bigl(\beta_{(x,a)}^x\beta_{(y,b)}^y\bigr)
    =\frac{\Tr\bigl(E_a^xD(F_b^y)^TD\bigr)}{(\Tr D)^2}=0.
\]
All other conditions in~\eqref{eq:zero-error-code} follow from the block supports. The realization in Section~\ref{sec:zero-error-preliminaries} gives this code using $\ket{\chi_D}$.

\emph{From a saturating zero-error code to a perfect strategy.} Consider any $q$-message code with receiver state $\rho$ and operators $\beta_v^i$. Apply Lemma~\ref{lem:clique-partition-saturation} and set
\[
    \gamma_{(z,c)}:=\sum_{i=1}^q\beta_{(z,c)}^i.
\]
For every question $z$, equation~\eqref{eq:clique-partition-saturation} gives
\begin{equation}
    \sum_{c:(z,c)\in C_z}\gamma_{(z,c)}
    =\sum_i\sigma_{i,z}=\rho.
\label{eq:aggregated-code-completeness}
\end{equation}
For a losing tuple $(x,a,y,b)$, every product $\beta_{(x,a)}^i\beta_{(y,b)}^j$ vanishes. If $i\neq j$, this follows from~\eqref{eq:zero-error-code} and positivity. If $i=j$, the factors are supported on $P_{i,x}$ and $P_{i,y}$, which are orthogonal because $x$ and $y$ are distinct tagged questions. Thus
\begin{equation}
    \gamma_{(x,a)}\gamma_{(y,b)}=0
    \qquad\bigl(V(a,b\mid x,y)=0\bigr).
\label{eq:aggregated-code-orthogonality}
\end{equation}
Choose an eigenbasis of $\rho$ and define
\[
    \widehat E_a^x=\rho^{-1/2}\gamma_{(x,a)}\rho^{-1/2},
    \qquad
    (\widehat F_b^y)^T=\rho^{-1/2}\gamma_{(y,b)}\rho^{-1/2},
    \qquad
    D':=\frac{\rho}{\sqrt{\Tr(\rho^2)}}.
\]
These measurement effects are positive and sum to the identity by~\eqref{eq:aggregated-code-completeness}. Moreover, for each losing tuple,
\[
    \Tr\bigl(\widehat E_a^xD'(\widehat F_b^y)^TD'\bigr)
    =\frac{\Tr\bigl(\gamma_{(x,a)}\gamma_{(y,b)}\bigr)}
           {\Tr(\rho^2)}=0.
\]
They therefore give a perfect strategy with Schmidt operator $D'$. If the code uses $\ket{\chi_D}$, then $\rho=D/\Tr D$ and $D'=D$, proving the desired converse.

\end{proof}

The next corollary follows directly from Theorem~\ref{thm:game-channel-transfer}.

\begin{corollary}[Optimal transmission and perfect play]
\label{cor:game-channel-capacities}
\[
\begin{aligned}
    \alpha^*(\Gamma_G)=q
    &\iff G\text{ has a perfect finite-dimensional strategy},\\
    \alpha^*_{\mathrm{ME}}(\Gamma_G)=q
    &\iff \begin{aligned}[t]
        &G\text{ has a perfect finite-dimensional}\\
        &\text{maximally entangled strategy}.
    \end{aligned}
\end{aligned}
\]
\end{corollary}

The correspondence in Theorem~\ref{thm:game-channel-transfer} also transfers state certification. If perfect play in $G$ certifies $\ket{\psi_D}$, then transmitting $q$ messages with zero error certifies $\ket{\chi_D}$, in both cases for finite-dimensional pure shared resources and up to local isometries and an auxiliary pure state.

Applied to Flagged Magic-Square, the reduction gives the following channel separation.

\begin{corollary}[Flagged Magic-Square channel]
\label{cor:FMS-zero-error}
There is a $50$-input classical channel $\mathcal N_{\mathrm{FMS}}$ whose confusability graph satisfies
\[
    \alpha^*(\Gamma_{\mathrm{FMS}})=10,
    \qquad
    \alpha^*_{\mathrm{ME}}(\Gamma_{\mathrm{FMS}})\leq9.
\]
\end{corollary}

An optimal code uses the state
\begin{equation}
    \ket{\chi_{\mathrm{FMS}}}
      =\frac{1}{\sqrt{1+4\sqrt2}}
       \left(
          \ket{\bot,\bot}
          +2^{1/4}\sum_{z\in\{0,1\}^2}\ket{z,z}
       \right).
\label{eq:FMS-zero-error-state}
\end{equation}
\clearpage
\appendix
\section*{Appendices}
\addcontentsline{toc}{section}{Appendices}

\section{Operator lemmas and Magic-Square rigidity}\label{app:operator-rigidity}

\subsection{Elementary operator lemmas}\label{app:operator-lemmas}

The following facts are used in the proof of Theorem~\ref{thm:FMS-separation} and in Section~\ref{app:tracial-MS-rigidity}.

\begin{lemma}[Component matching]\label{lem:component-matching}
Let $\{E_\alpha\}_{\alpha\in\mathcal A}$ and $\{F_\beta\}_{\beta\in\mathcal B}$ be POVMs on the same Hilbert space.  Let $S\subseteq\mathcal A$ and $T\subseteq\mathcal B$.  Suppose
\[
E_\alpha F_\beta=0
\quad\text{whenever}\quad
(\alpha\in S,\beta\notin T)
\text{ or }
(\alpha\notin S,\beta\in T).
\]
Then
\[
\sum_{\alpha\in S}E_\alpha
 =\sum_{\beta\in T}F_\beta
 =:P,
\]
and $P$ is a projection.
\end{lemma}

\begin{proof}
Set $E_S=\sum_{\alpha\in S}E_\alpha$ and $F_T=\sum_{\beta\in T}F_\beta$.  Completeness and the cross-component zero products give
\[
E_S=E_SF_T,
\qquad
F_T=E_SF_T.
\]
Thus $E_S=F_T=P$, and $P=P^2$.
\end{proof}

\begin{lemma}[Binary matching]\label{lem:binary-matching}
Let $E^+,E^-,F^+,F^-\succeq0$ satisfy
\[
E^++E^-=F^++F^-=I.
\]
Fix $s\in\T$.  If
\[
E^aF^b=0
\qquad\text{whenever }ab=-s,
\]
then all four effects are projections and
\[
E^a=F^{sa}
\qquad(a\in\T).
\]
\end{lemma}

\begin{proof}
For $s=+1$, the forbidden pairs are $(+,-)$ and $(-,+)$.  Hence
\[
E^+=E^+(F^++F^-)=E^+F^+,
\qquad
F^+=(E^++E^-)F^+=E^+F^+.
\]
Therefore $E^+=F^+=(E^+)^2$.  The same argument applies to the minus effects.  For $s=-1$, the identical calculation pairs $E^+$ with $F^-$ and $E^-$ with $F^+$.
\end{proof}

\subsection{Proof of Magic-Square rigidity}\label{app:tracial-MS-rigidity}

\begin{proof}[Proof of Proposition~\ref{prop:tracial-MS-rigidity}]
For each cell $(i,j)$ and sign $a\in\T$, define the coarse effects
\[
P_{ij}^a:=\sum_{s:s_j=a}A_s^i,
\qquad
Q_{ij}^a:=\sum_{t:t_i=a}B_t^j.
\]
Equation~\eqref{eq:tracial-MS-zero} gives
\[
P_{ij}^+Q_{ij}^-=P_{ij}^-Q_{ij}^+=0.
\]
Lemma~\ref{lem:binary-matching} implies
\[
P_{ij}^a=Q_{ij}^a,
\]
and these common operators are complementary projections.  Set
\[
M_{ij}:=P_{ij}^+-P_{ij}^-.
\]
Each $M_{ij}$ is a self-adjoint unitary.

For a fixed row $i$, every fine effect $A_s^i$ is supported on $P_{ij}^{s_j}$ for each $j$.  Consequently,
\[
P_{ij}^{a}P_{ik}^{b}
 =\sum_{s:s_j=a,\ s_k=b}A_s^i
 =P_{ik}^{b}P_{ij}^{a}.
\]
Thus the three row observables commute, and their joint spectral projections are the original effects $A_s^i$. In particular, the row POVM is projective.  The same argument applies to every column.  Since only valid parity assignments occur,
\begin{equation}
M_{i1}M_{i2}M_{i3}=I,
\qquad
M_{1j}M_{2j}M_{3j}=\kappa_j I.
\label{eq:tracial-MS-relations}
\end{equation}
Here $\kappa_j$ are the column parities defined in~\eqref{eq:column-parities}.

Set
\[
Z_1=M_{11},
\quad Z_2=M_{12},
\quad X_2=M_{21},
\quad X_1=M_{22}.
\]
The row and column commutation relations give
\begin{equation}
[Z_1,Z_2]=[X_1,X_2]=[Z_1,X_2]=[Z_2,X_1]=0.
\label{eq:tracial-cross-comm}
\end{equation}
The product relations determine
\[
M_{13}=Z_1Z_2,
\quad M_{23}=X_2X_1,
\quad M_{31}=Z_1X_2,
\quad M_{32}=Z_2X_1,
\]
and the third-column relation gives
\[
M_{33}=-(Z_1Z_2)(X_2X_1).
\]
Substituting these identities into the third-row relation in \eqref{eq:tracial-MS-relations}, and using only \eqref{eq:tracial-cross-comm}, yields
\[
Z_1X_1Z_1X_1=-I,
\]
so
\begin{equation}
Z_1X_1=-X_1Z_1.
\label{eq:tracial-first-anti}
\end{equation}
Moreover, $M_{13}$ and $M_{23}$ commute because they lie in the same column.  Hence
\[
(Z_1Z_2)(X_2X_1)=(X_2X_1)(Z_1Z_2).
\]
Using \eqref{eq:tracial-cross-comm} and \eqref{eq:tracial-first-anti}, cancellation gives
\begin{equation}
Z_2X_2=-X_2Z_2.
\label{eq:tracial-second-anti}
\end{equation}

A pair of anticommuting self-adjoint unitaries is unitarily equivalent to $Z,X$ on a qubit tensored with an identity: decompose into the $\pm1$ eigenspaces of the first unitary and use the second to identify the two eigenspaces.  Applying this first to $Z_1,X_1$, and then inside their commutant to $Z_2,X_2$, gives
\[
\mathcal K_0\cong\C^2\otimes\C^2\otimes\mathcal K
\]
with
\[
Z_1=Z\otimes I\otimes I,
\quad X_1=X\otimes I\otimes I,
\quad Z_2=I\otimes Z\otimes I,
\quad X_2=I\otimes X\otimes I.
\]
All nine $M_{ij}$ are therefore exactly the observables in \eqref{eq:magic-square}, tensored with $I_{\mathcal K}$.  Their joint spectral projections are the canonical row and column outcome projections, also tensored with $I_{\mathcal K}$.
\end{proof}

\section{Exact self-testing}\label{app:self-testing}

This appendix shows that the Flagged Magic-Square game self-tests the state in \eqref{eq:honest-state} and the measurements of Section~\ref{sec:honest-strategy}.

\subsection{Statement of the self-test}\label{app:self-test-statement}

Let $M_a^q$ denote Alice's ideal measurement effect for question $q$ and answer $a$ in Section~\ref{sec:honest-strategy}.  Let $N_b^r$ denote the corresponding unconjugated ideal effect on Bob's side, so Bob's physical ideal effect is $\overline{N_b^r}$.

\begin{theorem}[Exact self-test]\label{thm:self-test}
Let
\[
\bigl(\ket\Psi,\{A_a^q\}_{a,q},\{B_b^r\}_{b,r}\bigr)
\]
be a finite-dimensional perfect strategy for $G_{\mathrm{FMS}}$ with a pure shared state.  Let $\mathcal H_A^0$ and $\mathcal H_B^0$ be the Schmidt supports of $\ket\Psi$.  Then there exist a finite-dimensional Hilbert space $\mathcal J$, a unit vector $\ket{\mathrm{aux}}\in\mathcal J\otimes\mathcal J$, and unitary identifications
\[
V_A:\mathcal H_A^0\xrightarrow{\ \cong\ }\mathcal H\otimes\mathcal J,
\qquad
V_B:\mathcal H_B^0\xrightarrow{\ \cong\ }\mathcal H\otimes\mathcal J
\]
such that
\begin{equation}
(V_A\otimes V_B)\ket\Psi
 =\ket\psi\otimes\ket{\mathrm{aux}},
\label{eq:self-tested-state}
\end{equation}
where $\ket\psi$ is the state in \eqref{eq:honest-state}.  Moreover, after compressing all measurement effects to the Schmidt supports,
\begin{equation}
V_A A_a^q V_A^*=M_a^q\otimes I_{\mathcal J},
\qquad
V_B B_b^r V_B^*=\overline{N_b^r}\otimes I_{\mathcal J}
\label{eq:self-tested-measurements}
\end{equation}
for every question and answer.  Thus $G_{\mathrm{FMS}}$ exactly self-tests both the state and the measurements of the strategy in Section~\ref{sec:honest-strategy}.
\end{theorem}

\begin{remark}
Theorem~\ref{thm:self-test} immediately gives another proof that no perfect maximally entangled strategy exists: the reduced state of $\ket\psi$ is $\frac19(1\oplus2I_4)$, so tensoring it with any nonzero auxiliary reduced state cannot produce a scalar multiple of the identity on its support.
\end{remark}

\subsection{Weighted matching lemmas}\label{app:weighted-matching}

We recall elementary operator lemmas~\cite{Mancinska2014} and Magic-Square rigidity in the Schmidt-operator form used below, including short proofs for completeness.

\begin{lemma}[Weighted binary matching]\label{lem:weighted-binary}
Let $D>0$, and let $E^+,E^-,F^+,F^-\succeq0$ satisfy
\[
E^++E^-=F^++F^-=I.
\]
Fix $s\in\T$.  If
\[
E^aDF^b=0
\qquad\text{whenever }ab=-s,
\]
then all four effects are projections commuting with $D$, and
\[
E^a=F^{sa}
\qquad(a\in\T).
\]
\end{lemma}

\begin{proof}
It suffices to treat $s=+1$, since the case $s=-1$ follows by exchanging $F^+$ and $F^-$.  The forbidden pairs give
\[
E^+D=E^+D(F^++F^-)=E^+DF^+,
\qquad
DF^+=(E^++E^-)DF^+=E^+DF^+.
\]
Hence $E^+D=DF^+$.  Taking adjoints gives $DE^+=F^+D$, and therefore
\[
E^+D^2=DF^+D=D^2E^+.
\]
Thus $E^+$ commutes with $D$ by functional calculus, and $E^+D=DF^+$ then gives $E^+=F^+$.  Finally,
\[
0=E^+D(I-F^+)=D E^+(I-E^+),
\]
so invertibility of $D$ gives $(E^+)^2=E^+$.  The same conclusions hold for the complementary effects.
\end{proof}

\begin{lemma}[Weighted component matching]\label{lem:weighted-component}
Let $D>0$, let $\{E_\alpha\}_{\alpha\in\mathcal A}$ and $\{F_\beta\}_{\beta\in\mathcal B}$ be POVMs on the same finite-dimensional Hilbert space, and let $S\subseteq\mathcal A$, $T\subseteq\mathcal B$.  Suppose
\[
E_\alpha D F_\beta=0
\quad\text{whenever}\quad
(\alpha\in S,\beta\notin T)
\text{ or }
(\alpha\notin S,\beta\in T).
\]
Then
\[
\sum_{\alpha\in S}E_\alpha
 =\sum_{\beta\in T}F_\beta
 =:P,
\]
and $P$ is a projection commuting with $D$.
\end{lemma}

\begin{proof}
Coarse-grain the two POVMs to
\[
E_S:=\sum_{\alpha\in S}E_\alpha,
\qquad
F_T:=\sum_{\beta\in T}F_\beta.
\]
The hypotheses are exactly the losing relations for perfectly matching binary measurements.  Lemma~\ref{lem:weighted-binary} with $s=+1$ gives the result.
\end{proof}

\subsection{Schmidt-form Magic-Square rigidity}\label{app:schmidt-MS-rigidity}

We use the following reformulation of Magic-Square rigidity~\cite{WuBancalMcKagueScarani2016,ColadangeloStark}.

\begin{proposition}[Schmidt-form Magic-Square rigidity]\label{prop:schmidt-MS-rigidity}
Let $D>0$ act on a finite-dimensional Hilbert space $\mathcal K_0$. Let $\{A_s^i\}_s$ and $\{B_t^j\}_t$, $i,j=1,2,3$, be POVMs on $\mathcal K_0$, with outcomes indexed by the valid Magic-Square answers to $R_i$ and $C_j$, respectively (Section~\ref{sec:MS-notation}). Suppose
\begin{equation}
A_s^i D B_t^j=0
\qquad\text{whenever }s_j\ne t_i.
\label{eq:schmidt-MS-zero}
\end{equation}
Then all six POVMs are projective.  Moreover, there exist a Hilbert space $\mathcal K$, a positive definite operator $D_{\mathcal K}$ on $\mathcal K$, and a unitary identification
\[
\mathcal K_0\cong\C^2\otimes\C^2\otimes\mathcal K
\]
under which the six measurements are the joint-eigenbasis measurements of \eqref{eq:magic-square}, tensored with $I_{\mathcal K}$, and
\[
D=I_4\otimes D_{\mathcal K}.
\]
\end{proposition}

\begin{proof}
For each cell $(i,j)$ and sign $a\in\T$, define
\[
P_{ij}^a:=\sum_{s:s_j=a}A_s^i,
\qquad
Q_{ij}^a:=\sum_{t:t_i=a}B_t^j.
\]
Equation~\eqref{eq:schmidt-MS-zero} gives
\[
P_{ij}^+D Q_{ij}^-=P_{ij}^-D Q_{ij}^+=0.
\]
Lemma~\ref{lem:weighted-binary} implies that $P_{ij}^a=Q_{ij}^a$ and that these common operators are complementary projections commuting with $D$.  From this point, the proof of Proposition~\ref{prop:tracial-MS-rigidity} in Section~\ref{app:tracial-MS-rigidity} applies verbatim: the row and column POVMs are projective, and their observables are unitarily equivalent to the Mermin--Peres square tensored with an identity.

Every resulting cell observable $M_{ij}=P_{ij}^+-P_{ij}^-$ commutes with $D$.  These observables generate $M_4(\C)\otimes I_{\mathcal K}$, whose commutant is $I_4\otimes\mathcal B(\mathcal K)$.  Hence $D=I_4\otimes D_{\mathcal K}$ for a positive definite $D_{\mathcal K}$.
\end{proof}

\subsection{Proof of the exact self-test}\label{app:self-test-proof}

\begin{proof}[Proof of Theorem~\ref{thm:self-test}]
Compress every measurement effect to the corresponding Schmidt support and reuse the same notation.  The two supports have the same dimension.  After choosing Schmidt bases, identify both with a Hilbert space $\mathcal K_0$ and write
\begin{equation}
\ket\Psi=\operatorname{vec}(D),
\qquad
D>0,
\qquad
\Tr(D^2)=1.
\label{eq:general-Schmidt-operator}
\end{equation}
Transpose Bob's effects in the Schmidt basis, again denoting the transposed effects by $B_b^r$.  For an Alice effect $E$ and a transposed Bob effect $F$,
\begin{equation}
\Pr(E,F)=\Tr(D E D F).
\label{eq:weighted-probability}
\end{equation}
If this probability is zero, then
\[
0=\Tr(D E D F)=\bigl\|E^{1/2}D F^{1/2}\bigr\|_2^2,
\]
so
\begin{equation}
E D F=F D E=0.
\label{eq:weighted-zero}
\end{equation}
We use these weighted zero-product relations throughout.

\paragraph{Step 1: flag and Magic-Square sectors.}
Let $G_i=A_{\bot}^{R_i}$ and $H_j=B_{\bot}^{C_j}$.  For every $(i,j)$, the flag and nonflag outcomes form perfectly matching binary measurements.  Lemma~\ref{lem:weighted-binary} therefore gives
\[
G_i=H_j=:P^\perp
\qquad\text{for all }i,j,
\]
where $P^\perp$ is the common flag projection and commutes with $D$.  Set
\[
P:=I-P^\perp,
\qquad
\mathcal L:=P^\perp\mathcal K_0,
\qquad
\mathcal K_0^{\mathrm{MS}}:=P\mathcal K_0.
\]
The nonflag row and column effects are supported on $P$ and, with the restriction of $D$ to this corner, satisfy Proposition~\ref{prop:schmidt-MS-rigidity}.  After a common unitary identification,
\begin{equation}
\mathcal K_0
 \cong
\mathcal L\oplus\bigl(\C^2\otimes\C^2\otimes\mathcal K\bigr),
\qquad
D=D_{\mathrm f}\oplus\bigl(I_4\otimes D_{\mathrm m}\bigr),
\label{eq:state-sector-decomposition}
\end{equation}
where $D_{\mathrm f}>0$ on $\mathcal L$ and $D_{\mathrm m}>0$ on $\mathcal K$.  All nonflag row and column measurements are the canonical Magic-Square measurements tensored with $I_{\mathcal K}$.

For an outcome $\gamma$ of $C_1$ and an outcome $\rho$ of $R_1$, write $\Pi_\gamma$ and $\Pi_\rho$ for the corresponding canonical projections in the nonflag sector.  In particular,
\begin{equation}
\Pi_{\cans{x}}=\proj{c_x}\otimes I_{\mathcal K},
\qquad
\Pi_{\rans{y}}=\proj{r_y}\otimes I_{\mathcal K}.
\label{eq:selftest-distinguished-projections}
\end{equation}

\paragraph{Step 2: calibration constraints fix every bridge support.}
Let $E_x^+,E_x^-$ be Alice's sign effects for $T_x$, and let $F_y^+,F_y^-$ be Bob's sign effects for $U_y$.  Applying Lemma~\ref{lem:weighted-component} to the calibration constraints gives
\begin{equation}
E_x^++E_x^-=P^\perp+\Pi_{\cans{x}}=:S_x^A,
\qquad
F_y^++F_y^-=P^\perp+\Pi_{\rans{y}}=:S_y^B.
\label{eq:selftest-bridge-supports}
\end{equation}
The same lemma, applied to each singleton Magic-Square answer, shows that all remaining bridge effects are the corresponding calibrated outcome projections:
\begin{equation}
E_x^\gamma=\Pi_\gamma\quad(\gamma\ne\cans{x}),
\qquad
F_y^\rho=\Pi_\rho\quad(\rho\ne\rans{y}).
\label{eq:calibrated-singletons}
\end{equation}
Define the sign observables
\[
X_x:=E_x^+-E_x^-,
\qquad
Y_y:=F_y^+-F_y^-.
\]
They are self-adjoint contractions supported on $S_x^A$ and $S_y^B$, and
\begin{equation}
E_x^a=\frac{S_x^A+aX_x}{2},
\qquad
F_y^b=\frac{S_y^B+bY_y}{2}
\quad(a,b\in\T).
\label{eq:selftest-sign-decomposition}
\end{equation}

\paragraph{Step 3: the frustrated square fixes the bridge operators and the Schmidt ratio.}
Put
\[
s_{xy}:=(-1)^{xy},
\qquad
\alpha_{xy}:=\ip{c_x}{r_y}=\frac{s_{xy}}{\sqrt2}.
\]
For a losing sign pair, $b=-s_{xy}a$, and \eqref{eq:weighted-zero} gives
\[
(S_x^A+aX_x)D(S_y^B-s_{xy}aY_y)=0
\qquad(a\in\T).
\]
Adding and subtracting the two equations yields
\begin{align}
X_x D Y_y&=s_{xy}S_x^A D S_y^B,
\label{eq:weighted-core1}\\
X_x D S_y^B&=s_{xy}S_x^A D Y_y.
\label{eq:weighted-core2}
\end{align}

Identify $S_x^A\mathcal K_0$ with $\mathcal L\oplus\mathcal K$ by identifying $\ket{c_x}\otimes\mathcal K$ with $\mathcal K$, and identify $S_y^B\mathcal K_0$ with $\mathcal L\oplus\mathcal K$ using $\ket{r_y}\otimes\mathcal K$.  Write
\begin{equation}
X_x=
\begin{pmatrix}
\Gamma_x&U_x^*\\
U_x&\Delta_x
\end{pmatrix},
\qquad
Y_y=
\begin{pmatrix}
\Lambda_y&V_y^*\\
V_y&\Theta_y
\end{pmatrix},
\label{eq:bridge-blocks}
\end{equation}
where $U_x,V_y:\mathcal L\to\mathcal K$.  Under these identifications,
\[
S_x^A D S_y^B
=
\begin{pmatrix}
D_{\mathrm f}&0\\
0&\alpha_{xy}D_{\mathrm m}
\end{pmatrix}.
\]
Comparing blocks in \eqref{eq:weighted-core2} gives
\begin{align}
\Gamma_xD_{\mathrm f}&=s_{xy}D_{\mathrm f}\Lambda_y,
&\Delta_xD_{\mathrm m}&=s_{xy}D_{\mathrm m}\Theta_y,
\label{eq:diagonal-bridge-relations}\\
U_xD_{\mathrm f}&=\frac1{\sqrt2}D_{\mathrm m}V_y,
&U_x^*D_{\mathrm m}&=\sqrt2D_{\mathrm f}V_y^*.
\label{eq:offdiagonal-bridge-relations}
\end{align}
The four signs force
\begin{equation}
\Gamma_x=\Delta_x=\Lambda_y=\Theta_y=0
\qquad(x,y\in\{0,1\}).
\label{eq:zero-diagonal-bridge}
\end{equation}
For example, the first relation in \eqref{eq:diagonal-bridge-relations} gives $\Lambda_0=\Lambda_1$ from the two cells in the first row of the sign table and $\Lambda_0=-\Lambda_1$ from the two cells in the second row.  Invertibility of $D_{\mathrm f}$ then gives $\Gamma_x=\Lambda_y=0$.  The argument for $\Delta_x,\Theta_y$ is identical.

Equation~\eqref{eq:offdiagonal-bridge-relations} and invertibility of $D_{\mathrm f},D_{\mathrm m}$ show that $U_0=U_1=:U$ and $V_0=V_1=:V$.  Comparing the diagonal blocks in \eqref{eq:weighted-core1}, now using \eqref{eq:zero-diagonal-bridge}, gives
\begin{equation}
U^*D_{\mathrm m}V=\sqrt2D_{\mathrm f},
\qquad
UD_{\mathrm f}V^*=\frac1{\sqrt2}D_{\mathrm m}.
\label{eq:Schmidt-ratio-relations}
\end{equation}
These identities force $U$ and $V$ to be unitary.  Indeed,
\[
U^*U D_{\mathrm f}
 =\frac1{\sqrt2}U^*D_{\mathrm m}V
 =D_{\mathrm f},
\qquad
UU^*D_{\mathrm m}
 =\sqrt2UD_{\mathrm f}V^*
 =D_{\mathrm m},
\]
so $U^*U=I_{\mathcal L}$ and $UU^*=I_{\mathcal K}$.  Similarly, comparing \eqref{eq:offdiagonal-bridge-relations} with \eqref{eq:Schmidt-ratio-relations} gives $V^*V=I_{\mathcal L}$ and $VV^*=I_{\mathcal K}$.  In particular, $\mathcal L$ and $\mathcal K$ have the same dimension.

\paragraph{Step 4: extract the state and measurements.}
Let $\mathcal J:=\mathcal L$.  With the target identified as $\mathcal H\otimes\mathcal J$, define
\begin{equation}
W_A:=I_{\mathcal L}\oplus(I_4\otimes U^*),
\qquad
W_B:=I_{\mathcal L}\oplus(I_4\otimes V^*).
\label{eq:extracting-unitaries}
\end{equation}
The first summand is identified with $\C\ket\bot\otimes\mathcal J$ and the second with $(\C^2\otimes\C^2)\otimes\mathcal J$.  By \eqref{eq:Schmidt-ratio-relations},
\begin{equation}
W_A D W_B^*
 =D_{\mathrm f}\oplus\bigl(I_4\otimes U^*D_{\mathrm m}V\bigr)
 =D_0\otimes D_{\mathrm f}.
\label{eq:extracted-Schmidt-operator}
\end{equation}
Normalization of $D$ implies $9\Tr(D_{\mathrm f}^2)=1$.  Hence
\[
D_{\mathrm{aux}}:=3D_{\mathrm f}
\]
has Hilbert--Schmidt norm one.  Taking $V_A=W_A$, $V_B=\overline{W_B}$ in the original Schmidt bases, and $\ket{\mathrm{aux}}=\operatorname{vec}(D_{\mathrm{aux}})$, equation~\eqref{eq:extracted-Schmidt-operator} gives \eqref{eq:self-tested-state}, up to the canonical regrouping of tensor factors.

It remains to identify the measurements.  The row and column measurements are canonical by Proposition~\ref{prop:schmidt-MS-rigidity}, and \eqref{eq:calibrated-singletons} identifies every nonsign bridge outcome.  By \eqref{eq:zero-diagonal-bridge},
\[
X_x=
\begin{pmatrix}0&U^*\\ U&0\end{pmatrix}
\quad\text{on }\mathcal L\oplus(\ket{c_x}\otimes\mathcal K).
\]
Conjugating by $W_A$ turns this into
\[
\bigl(\ket\bot\!\bra{c_x}+\ket{c_x}\!\bra\bot\bigr)\otimes I_{\mathcal J}.
\]
Together with \eqref{eq:selftest-sign-decomposition}, this gives
\[
W_AE_x^aW_A^*=\proj{t_x^a}\otimes I_{\mathcal J}.
\]
The same argument with $Y_y,V$ gives the ideal $U_y$ effects in the transposed Bob picture.  Taking complex conjugates returns Bob's physical measurements.  This proves \eqref{eq:self-tested-measurements}.
\end{proof}
\section{Flagging suitable base games}
\label{app:flagged-base-games}

\paragraph{Suitable base games.}
Let $G$ be a finite bipartite game with question sets $X,Y$, fixed answer sets $\mathcal A,\mathcal B$, and a full-support question distribution. By adding answers that always lose if necessary, we may assume $|\mathcal A|=|\mathcal B|$. We call $G$ a \emph{suitable base game at $c$} if:
\begin{enumerate}
\item[(i)] $G$ is a pseudo-telepathy game that exactly self-tests a value-one projective strategy
\[
 \bigl(|\Phi_d\rangle,\{P^q_a\}_{q,a},
       \{\overline{Q^r_b}\}_{r,b}\bigr),\qquad
 |\Phi_d\rangle=\frac1{\sqrt d}\sum_{j=1}^d|j,j\rangle.
\]
The self-test is required for every finite-dimensional perfect strategy with a pure shared state and arbitrary POVMs. It includes both the state and the measurements on their Schmidt supports, up to local isometries and an auxiliary factor.

\item[(ii)] The canonical measurements have distinguished rank-one outcomes at questions $r_x\in Y$, $q_y\in X$:
\[
 Q^{r_x}_{b_x}=|w_x\rangle\langle w_x|,\qquad
 P^{q_y}_{a_y}=|v_y\rangle\langle v_y|,
 \qquad x,y\in\{0,1\},
\]
whose unit representatives can be phased so that
\begin{equation}
 \langle w_x|v_y\rangle=c\,s_{xy},\qquad
 s_{xy}=(-1)^{xy},\qquad 0<c<1.
 \label{eq:base-cross-overlaps}
\end{equation}
The questions $r_0,r_1$ may differ, as may $q_0,q_1$. Neither pair $w_0,w_1$ nor $v_0,v_1$ is required to be orthogonal.
\end{enumerate}
Here $v$ denotes Alice's canonical directions and $w$ Bob's unconjugated directions, as in \eqref{eq:general-flagging-overlaps}. The reversed ordering in \eqref{eq:base-cross-overlaps} is equivalent because the sign matrix is real and symmetric. For unit vectors, this overlap pattern forces $c\le 1/\sqrt2$, so $c<1$ imposes no additional restriction. Indeed, Cauchy--Schwarz gives
\[
 8c^2\le\|w_0+w_1\|^2+\|w_0-w_1\|^2=4.
\]

\paragraph{The flagged game.}
For fixed choices in (ii), form $G^{\mathrm{flag}}$ as in Section~\ref{sec:game}: add the flag answer $\bot$ to every base question, and add bridge questions $T_0,T_1$ for Alice and $U_0,U_1$ for Bob. On every base question pair, $(\bot,\bot)$ wins, exactly one flag loses, and two nonflag answers are tested by $G$. The bridge answer sets are
\[
 \{+1,-1\}\sqcup(\mathcal B\setminus\{b_x\}),\qquad
 \{+1,-1\}\sqcup(\mathcal A\setminus\{a_y\}).
\]
Use the calibration rules of Section~\ref{sec:game} at $(T_x,r_x)$ and $(q_y,U_y)$: signs match $\{\bot,b_x\}$ and $\{\bot,a_y\}$, respectively, and every nonsign bridge answer must equal the base answer. On $(T_x,U_y)$ retain the frustrated sign rule: two signs $a,b$ win exactly when $ab=s_{xy}$, and any nonsign answer wins automatically. All other cells are free. Use any full-support distribution on the enlarged question product.

On $\mathcal H_c=\mathbb C|\bot\rangle\oplus\mathbb C^d$, define
\begin{equation}
 D_c=1\oplus c^{-1}I_d,\qquad
 |\psi_c\rangle=
 \frac{\operatorname{vec}(D_c)}{\sqrt{1+dc^{-2}}}
 =\frac{c|\bot,\bot\rangle+\sum_{j=1}^d|j,j\rangle}
 {\sqrt{c^2+d}}.
 \label{eq:flagged-target}
\end{equation}
The ideal measurements are those of Section~\ref{sec:honest-strategy} with canonical core measurements $P,Q$. Each bridge replaces the flag and its selected outcome by rank-one projections onto the vectors
\[
 |t_x^a\rangle=\frac{|\bot\rangle+a|w_x\rangle}{\sqrt2},\qquad
 |u_y^b\rangle=\frac{|\bot\rangle+b|v_y\rangle}{\sqrt2},
 \qquad a,b\in\{+1,-1\},
\]
and retains the other effects of its calibrating measurement. Bob uses the complex conjugates of his ideal effects.

\begin{theorem}[Perfect play and tracial separation]
\label{thm:flagged-base-separation}
Let $G$ be suitable at $c$. Then $\omega^*(G^{\mathrm{flag}})=1$, and no finite-dimensional maximally entangled strategy wins $G^{\mathrm{flag}}$ perfectly. If, additionally, the self-testing rigidity extends to perfect tracial realizations of $G$---in every faithful Hilbert-space representation, their two POVM families are simultaneously equivalent, via a common unitary, to $P^q_a\otimes I_{\mathcal K}$ and $Q^r_b\otimes I_{\mathcal K}$, with arbitrary-dimensional multiplicity---then
\[
 \omega^*(G^{\mathrm{flag}})=1>
 \omega_{\mathrm{tr}}(G^{\mathrm{flag}})
 \geq\omega_{\mathrm{ME}}(G^{\mathrm{flag}}).
\]
\end{theorem}
\begin{proof}[Proof sketch]
The completeness calculation of Proposition~\ref{prop:perfect} applies. The canonical base strategy gives $P_a^qQ_b^r=0$ for every losing base pair, and calibration uses the mutually orthogonal projections of the corresponding canonical measurement. For the sign constraints, $\langle t_x^a|D_c|u_y^b\rangle=(1+ab\,s_{xy})/2$. Theorem~\ref{thm:flagged-base-self-test} below rules out perfect maximally entangled play: tensoring the unequal Schmidt coefficients of $|\psi_c\rangle$ with any nonzero auxiliary spectrum cannot make them all equal. Under the additional tracial hypothesis, repeat the proof of Theorem~\ref{thm:FMS-separation} in a faithful GNS representation, using that hypothesis in place of Proposition~\ref{prop:tracial-MS-rigidity} and replacing the overlap bound $1/\sqrt2$ in \eqref{eq:projection-overlap} by $c<1$. Step~3 eliminates the flag sector, and Step~4, using the nonzero cross-overlaps, eliminates the core sector. Attainment of the tracial value and $\omega_{\mathrm{ME}}\leq\omega_{\mathrm{tr}}$ (Section~\ref{sec:notation}) give the strict separation.
\end{proof}

\begin{theorem}[Exact state self-test]
\label{thm:flagged-base-self-test}
Let $G$ be suitable at $c$. Every finite-dimensional perfect strategy for $G^{\mathrm{flag}}$ with a pure shared state $|\Psi\rangle$ admits unitary identifications $V_A,V_B$ of its Schmidt supports with $\mathcal H_c\otimes\mathcal J$, for a finite-dimensional $\mathcal J$, such that
\[
 (V_A\otimes V_B)|\Psi\rangle
 =|\psi_c\rangle\otimes|\mathrm{aux}\rangle
\]
for a unit vector $|\mathrm{aux}\rangle\in\mathcal J\otimes\mathcal J$.
\end{theorem}
\begin{proof}
Compress to the Schmidt supports, identify them with $\mathcal H_0$, and transpose Bob's effects, as in Appendix~\ref{app:self-testing}. Thus $|\Psi\rangle=\operatorname{vec}(D)$ with $D>0$ and $\operatorname{Tr}(D^2)=1$, and every losing pair of effects satisfies $EDF=0$, by \eqref{eq:weighted-zero}.

\emph{Flag, core, and bridge supports.} Weighted binary matching (Lemma~\ref{lem:weighted-binary}), applied to every base question pair, identifies all flag effects with a projection $P^\perp$ commuting with $D$. If the nonflag corner is nonzero, its normalized state and nonflag effects give a perfect strategy for $G$. By (i), the self-testing isometries restrict to local unitaries on its Schmidt supports once the auxiliary state is restricted to its own Schmidt supports. Denote these unitaries by $Z_A,\overline{Z_B}$. They take the measurements to their canonical forms and, after choosing Schmidt bases for the auxiliary state, satisfy
\[
 Z_A D_{\mathrm{core}} Z_B^*=I_d\otimes D_m>0,
 \qquad D_{\mathrm{core}}=D|_{(I-P^\perp)\mathcal H_0}.
\]
Since $D_{\mathrm{core}}>0$, the polar unitary on the left is $Z_AZ_B^*$. Positivity therefore gives $Z_A=Z_B$. Thus, after a common unitary identification,
\begin{equation}
 \mathcal H_0=\mathcal L\oplus(\mathbb C^d\otimes\mathcal K),\qquad
 D=D_f\oplus(I_d\otimes D_m),
 \label{eq:flagged-sector-form}
\end{equation}
where $\mathcal L=P^\perp\mathcal H_0$, and $D_f,D_m$ are positive definite on their respective spaces. The nonflag effects are canonical. Zero summands are allowed provisionally. Writing $\Pi_{w_x}=|w_x\rangle\langle w_x|\otimes I_{\mathcal K}$ and $\Pi_{v_y}=|v_y\rangle\langle v_y|\otimes I_{\mathcal K}$ in the core, weighted component matching (Lemma~\ref{lem:weighted-component}), applied separately at each calibration cell, gives
\begin{equation}
 E_x^++E_x^-=P^\perp+\Pi_{w_x}=:S_x^A,\qquad
 F_y^++F_y^-=P^\perp+\Pi_{v_y}=:S_y^B.
 \label{eq:flagged-supports}
\end{equation}
Here $E_x^a,F_y^b$ are the sign effects. Set $X_x=E_x^+-E_x^-$ and $Y_y=F_y^+-F_y^-$. These self-adjoint contractions are supported on $S_x^A,S_y^B$.

\emph{The frustrated square.} Adding and subtracting the two losing relations $E_x^aDF_y^{-s_{xy}a}=0$ gives
\begin{equation}
 X_xDY_y=s_{xy}S_x^ADS_y^B,\qquad
 X_xDS_y^B=s_{xy}S_x^ADY_y.
 \label{eq:flagged-sign-identities}
\end{equation}
Identify the two bridge supports with $\mathcal L\oplus\mathcal K$ using $w_x$ and $v_y$, respectively. In these identifications write
\[
 X_x=\begin{pmatrix}\Gamma_x&U_x^*\\U_x&\Delta_x\end{pmatrix},\qquad
 Y_y=\begin{pmatrix}\Lambda_y&V_y^*\\V_y&\Theta_y\end{pmatrix},\qquad
 S_x^ADS_y^B=
 \begin{pmatrix}D_f&0\\0&c\,s_{xy}D_m\end{pmatrix},
\]
where $U_x,V_y:\mathcal L\to\mathcal K$. The second identity in \eqref{eq:flagged-sign-identities} yields
\begin{align}
 \Gamma_xD_f&=s_{xy}D_f\Lambda_y,
 &\Delta_xD_m&=s_{xy}D_m\Theta_y,
 \label{eq:flagged-diagonal-blocks}\\
 U_xD_f&=cD_mV_y,
 &U_x^*D_m&=c^{-1}D_fV_y^*.
 \label{eq:flagged-off-diagonal-blocks}
\end{align}
The two rows of the sign table in \eqref{eq:flagged-diagonal-blocks} imply both $\Lambda_0=\Lambda_1$ and $\Lambda_0=-\Lambda_1$, and likewise for $\Theta_0,\Theta_1$. Invertibility therefore gives $\Gamma_x=\Delta_x=\Lambda_y=\Theta_y=0$. Equation~\eqref{eq:flagged-off-diagonal-blocks} also gives $U_0=U_1=:U$ and $V_0=V_1=:V$. The first identity in \eqref{eq:flagged-sign-identities} now becomes
\begin{equation}
 U^*D_mV=c^{-1}D_f,\qquad UD_fV^*=cD_m.
 \label{eq:flagged-weight-transfer}
\end{equation}
Combining \eqref{eq:flagged-off-diagonal-blocks} and \eqref{eq:flagged-weight-transfer} gives
\[
 \begin{aligned}
 U^*UD_f&=D_f, & UU^*D_m&=D_m,\\
 D_fV^*V&=D_f, & D_mVV^*&=D_m.
 \end{aligned}
\]
Thus $U,V:\mathcal L\to\mathcal K$ are unitaries. In particular, $\mathcal L$ and $\mathcal K$ have the same dimension. Neither can vanish, since otherwise both vanish, contradicting $\operatorname{Tr}(D^2)=1$.

\emph{Extraction.} Set $\mathcal J=\mathcal L$ and, with codomain $\mathcal H_c\otimes\mathcal J$, define
\[
 W_A=I_{\mathcal L}\oplus(I_d\otimes U^*),\qquad
 W_B=I_{\mathcal L}\oplus(I_d\otimes V^*).
\]
By \eqref{eq:flagged-weight-transfer},
\[
 W_ADW_B^*
 =D_f\oplus(I_d\otimes U^*D_mV)
 =D_c\otimes D_f.
\]
Consequently $(1+dc^{-2})\operatorname{Tr}(D_f^2)=1$, so $|\mathrm{aux}\rangle=\sqrt{1+dc^{-2}}\operatorname{vec}(D_f)$ is a unit vector. The vectorization identity gives
\[
 (W_A\otimes\overline{W_B})\operatorname{vec}(D)
 =|\psi_c\rangle\otimes|\mathrm{aux}\rangle,
\]
up to the canonical regrouping of tensor factors. Composing with the earlier Schmidt-basis changes and common unitary identification gives $V_A,V_B$ as claimed.
\end{proof}
\section*{Acknowledgements}
The author thanks Olivier Lalonde and David Roberson for helpful conversations. Converstaions with ChatGPT 5.6 Sol have been used in this work. The author is supported by the European Union via an ERC grant (QInteract, Grant No 101078107).

\end{document}